\documentclass[11pt]{article}
\usepackage{geometry}
\usepackage{graphicx}
\usepackage{amsmath}
\usepackage{amssymb}
\usepackage{xspace}
\usepackage[tight]{subfigure}
\usepackage[numbers]{natbib}

\newtheorem{theorem}{Theorem}[section]

\newtheorem{lemma}[theorem]{Lemma}
\newtheorem{proposition}[theorem]{Proposition}

\def\squarebox#1{\hbox to #1{\hfill\vbox to #1{\vfill}}}
\newcommand{\qed}{\hspace*{\fill}\vbox{\hrule\hbox{\vrule\squarebox{.667em}\vrule}\hrule}\smallskip}
\newenvironment{proof}{\noindent{\bf Proof:~~}}{\(\qed\)}

\newcommand{\xhdr}[1]{\smallskip {\bf #1.}}

\def\eps{{\varepsilon}}

\newcommand{\omt}[1]{}
\newcommand{\Xomit}[1]{}
\newcommand{\supproof}[1]{}
\newcommand{\supproofeq}[1]{}

\begin{document}
\title{Time-Inconsistent Planning: A Computational Problem in Behavioral Economics}
  
  \author{  
  Jon Kleinberg
  \thanks{
  Cornell University, Ithaca NY.
  Email: kleinber@cs.cornell.edu.
  Supported in part by
	a Simons Investigator Award,
	a Google Research Grant,
	an ARO MURI grant,
	and NSF grants
	IIS-0910664, 
	CCF-0910940, 
	and IIS-1016099. 
  }
   \and 
  Sigal Oren
  \thanks{
  Hebrew University and Microsoft Research, Israel.
  Email: sigalo@cs.huji.ac.il. 
  Supported in part by an I-CORE fellowship and a Microsoft Research Fellowship.
  }
  }
  
  
  
  \begin{titlepage}
  \maketitle

\begin{abstract}

In many settings, people exhibit behavior that is inconsistent across time ---
we allocate a block of time to get work done and then procrastinate,
or put effort into a project and then later fail to complete it.
An active line of research in behavioral economics and related fields
has developed and analyzed models for this type of time-inconsistent behavior.

Here we propose a graph-theoretic model of tasks and goals, in which
dependencies among actions are represented by a directed graph, and
a time-inconsistent agent constructs a path through this graph.
We first show how instances of this path-finding problem on different input 
graphs can reconstruct a wide range of qualitative phenomena observed
in the literature on time-inconsistency, including procrastination,
abandonment of long-range tasks, and the benefits of reduced sets
of choices.
We then explore a set of analyses that quantify over the set of all
graphs; among other results, we find that in any graph, there
can be only polynomially many distinct forms of time-inconsistent behavior;
and any graph in which a
time-inconsistent agent incurs significantly more cost than an 
optimal agent must contain a large ``procrastination'' structure
as a minor.
Finally, we use this graph-theoretic model to explore ways in which
tasks can be designed to help motivate agents to reach designated goals.

\end{abstract}



\end{titlepage}


\section{Introduction}

A fundamental issue in behavioral economics --- and in the modeling
of individual decision-making more generally --- is to understand
the effects of decisions that are inconsistent over time.
Examples of such inconsistency are widespread in everyday life:
we make plans for completing a task but then procrastinate;
we put work into getting a project partially done but then abandon it;
we pay for gym memberships but then fail to make use of them.
In addition to analyzing and modeling these effects, there has
been increasing interest in 
incorporating them into the design of policies and incentive systems
in domains that range from health to personal finance.

These types of situations have a recurring structure: a person makes
a plan at a given point in time for something they will do in the future
(finishing homework, exercising, paying off a loan), 
but at a later point in time they fail to follow through on the plan.
Sometimes this failure is the result of unforeseen circumstances that
render the plan invalid --- a person might join a gym but then break
their leg and be unable to exercise --- but in many cases the plan
is abandoned even though the circumstances are essentially the same
as they were at the moment the plan was made.
This presents a challenge to any model of planning based on 
optimizing a utility function that is consistent over time: 
in an optimization framework, the plan must have been an optimal
choice at the outset, but later it was optimal to abandon it.
A line of work in the economics literature has thus investigated the properties
of planning with objective functions that vary over time in certain
natural and structured ways.

\xhdr{A Basic Example and Model}
To introduce these models, it is useful to briefly describe an example
due to George Akerlof \cite{akerlof-procrastination}, with the
technical details adapted slightly for the discussion here.
(The story will be familiar to readers who know Akerlof's paper;
we cover it in some detail because it will motivate a useful
and recurring construction later in the work.)
Imagine a decision-making agent --- Akerlof himself, in his story ---
who needs to ship a package sometime during one of the next $n$ days,
labeled $t = 1, 2, \ldots, n$, and must decide on which day $t$ to do so. 
Each day that the package has not reached its destination results in a
cost of $x$ (per day), due to the lack of use of the package's contents.
If we suppose that the package takes a constant $h$ days in transit, this means
a cost of $(t + h) x$ if it is shipped on day $t$.
Also, shipping the package is an elaborate operation that will result in 
one-time cost of $c$, due to the amount of time involved in getting it sent 
out.  
The package must be shipped during one of the $n$ specified days.

What is the optimal plan for shipping the package?  Clearly the cost $c$
will be incurred exactly once regardless of the day on which it is shipped,
and there will also be a cost of $(t + h) x$ if it is shipped on day $t$.
Thus we are minimizing $c + tx + hx$ subject to $1 \leq t \leq n$;
since everything but $t$ is a constant in the objective function,
the cost is clearly minimized by setting $t = 1$.  In other words,
the agent should ship the package right away.

But in Akerlof's story, he did something that should be familiar from
one's own everyday experience: he procrastinated.  Although there were
no unexpected changes to the trade-offs involved in shipping the package,
when each new day arrived there seemed to be other things that were more 
crucial than sending it out that day, and so each day he resolved 
that he would instead send it tomorrow.  The result was that the package
was not sent out until the end of the time period.
(In fact, he sent it a few months into the time period 
once something unexpected did happen to change the cost structure ---
a friend offered to send it for him as part of a larger shipment ---
though this wrinkle is not crucial for the story.)

There is a natural way to model an agent's decision to procrastinate,
using the notion of {\em present bias} --- the tendency to view
costs and benefits that are incurred at the present moment to be
more salient than those incurred in the future.
In particular, 
suppose that for a constant $b > 1$, costs that one must incur
in the current time period are increased by a factor of $b$
in one's evaluation.\footnote{Note that there is no time-discounting
in this example, so the factor of $b$ is only applied to the present
time period, while all future time periods are treated equally.
We will return to the issue of discounting shortly.}
Then in Akerlof's example, when the agent on day $t$ is
considering the decision to send the package, the cost of sending
it on day $t$ is $b c + (t + h)x$, while the cost of sending it on 
day $(t + 1)$ is $c + (t + 1 + h) x$.
The difference between these two costs is $(b - 1) c - x$, and so if
$(b - 1) c > x$, the agent will decide on each day $t$ that the optimal
plan is to wait until day $t + 1$; things will continue this way until
day $n$, when waiting is no longer an option and the package must be sent.

\xhdr{Quasi-Hyperbolic Discounting}
Building on considerations such as those above, and others in earlier
work in economics \cite{strotz-time-inconsist,pollak-time-inconsist},
a significant amount of work has developed around a model of
time-inconsistency known as {\em quasi-hyperbolic discounting}
\cite{laibson-quasi-hyperbolic}.
In this model, parametrized by quantities $\beta, \delta \leq 1$,
a cost or reward of value $c$ that will be realized 
at a point $t \geq 1$ time units into the future is evaluated as having
a present value of $\beta \delta^t c$.
(In other words, values at time $t$
are discounted by a factor of $\beta \delta^t$.)
With $\beta = 1$ this is the standard functional form for 
exponential discounting, but when $\beta < 1$ the function
captures present bias as well: values in the present time period
are scaled up by $\beta^{-1}$ relative to all other periods.
(In what follows, we will consistently use $b$ to denote $\beta^{-1}$.)

Research on this $(\beta,\delta)$-model of discounting has been
extensive, and has proceeded in a wide variety of directions;
see Frederick et al. \cite{frederick-time-inconsist-surv} for a review.
To keep our analysis clearly delineated in scope, we make certain
decisions at the outset relative to the full
range of possible research questions:
we focus on a model of agents who are {\em naive}, in that they
do not take their own time-inconsistency into account when planning;
we do not attempt to derive the $(\beta,\delta)$-model from more
primitive assumptions but rather take it as a self-contained description 
of the agent's observed behavior; 
and we discuss the case of $\delta = 1$
so as to focus attention on the present-bias parameter $\beta$.
Note that the initial Akerlof example has all these properties;
it is essentially described in terms of the $(\beta,\delta)$-model
with an agent who is naive about his own time-inconsistency, with
$\delta = 1$, and with $\beta = b^{-1}$ (using the parameter
$b$ from that discussion).

Our starting point in this paper is to think about some of the 
qualitative predictions of the $(\beta,\delta)$-model, and how
to analyze them in a unified framework.  
In particular, research in behavioral economics has shown how agents making
plans in this model can exhibit the following behaviors.
\begin{enumerate}
\item {\em Procrastination}, as discussed above.
\item {\em Abandonment} of long-range tasks, in which a person starts
on a multi-stage project but abandons it in the middle, even though
the underlying costs and benefits of the project have remained
essentially unchanged \cite{odonoghue-long-term}.\footnote{For 
purposes of our discussion, we
distinguish abandonment of a task from the type of procrastination 
exhibited by Akerlof's example, in which the task is eventually finished,
but at a much higher cost due to the effect of procrastination.}
\item The benefits of {\em choice reduction}, in which reducing the set of
options available to an agent can actually help them reach a goal more
efficiently 
\cite{odonoghue-now-or-later,kaur-self-control-aer}.
A canonical example is the way in which imposing a
deadline can help people complete a task that might not get finished
in the absence of a deadline 
\cite{ariely-deadlines02}.
\end{enumerate}

These consequences of time-inconsistency, as well as a number
of others, have in general each required their own separate 
and sometimes quite intricate modeling efforts.
It is natural to ask whether there might instead be a single framework for
representing tasks and goals in which all of these effects could instead
emerge ``mechanically,'' each just as a different instance of the same generic 
computational problem.
With such a framework, it would become possible to search for worst-case
guarantees, by quantifying over all instances, and to talk about
designing or modifying given task structures to induce certain
desired behaviors.

\xhdr{The present work: A graph-theoretic model}
Here we propose such a framework, using a graph-theoretic formulation.
We consider an agent with present-bias parameter $\beta$ who must
construct a path in a directed acyclic graph $G$ with edge costs,
from a designated start node $s$ to a designated target node $t$. 
We will call such a structure a {\em task graph}.
Informally, 
the nodes of the task graph represent states of intermediate progress
toward the goal $t$, and the edges represent transitions between them.
Directed graphs have been shown to have considerable expressive
power for planning problems in the artificial intelligence literature
\cite{russell-norvig-book}; this provides evidence for the robustness
of a graph-based approach in representing these types of decision environments.
Our concerns in this work, however, are quite distinct
from the set of graph-based planning problems in artificial intelligence,
since our aim is to study the particular consequences of time-inconsistency
in these domains.

A sample instance of this problem is depicted in Figure \ref{fig:intro-ex01},
with the costs drawn on the edges.
When the agent is standing at a node $v$, it determines the minimum cost
of a path from $v$ to $t$, but it does so using its present-biased evaluation
of costs: the cost of the first edge on the path (starting from $v$) is 
evaluated according to the true cost, and all subsequent edges have
their costs reduced by $\beta$.
If the agent chooses path $P$, it follows just the first edge $(v,w)$ of $P$,
and then it re-evaluates which path to follow
using this same present-biased evaluation but now from the node $w$.
In this way, the agent iteratively constructs a path from $s$ to $t$.

In the next section we will show how our graph-theoretic model 
easily captures time-inconsistency phenomena including
procrastination, abandonment, and choice reduction.
But to make the definitions concrete, it is useful to work through
the agent's computation on the graph depicted in Figure \ref{fig:intro-ex01}.
In Figure \ref{fig:intro-ex01}, an agent that has a present-bias parameter of $\beta = 1/2$ needs to go from $s$ to $t$.
From $s$, the agent evaluates the path $s$-$a$-$b$-$t$ as having cost 
$16 + 2 \beta + 2 \beta = 18$, the path 
$s$-$c$-$d$-$t$ as having cost $8 + 8 \beta + 8 \beta = 16$, and the path 
$s$-$c$-$e$-$t$ as having cost $8 + 2 \beta + 16 \beta = 17$.
Thus the agent traverses the edge $(s,c)$ and ends up at $c$.
From $c$, the agent now evaluates the path $c$-$d$-$t$ as having cost
$8 + 8 \beta = 12$ and the path $c$-$e$-$t$ as having cost
$2 + 16 \beta = 10$, and so the agent traverses the edge 
$(c,e)$ and then (having no further choices) continues on the edge $(e,t)$.

\begin{figure}[t]
\begin{center}
\includegraphics[width=2.50in]{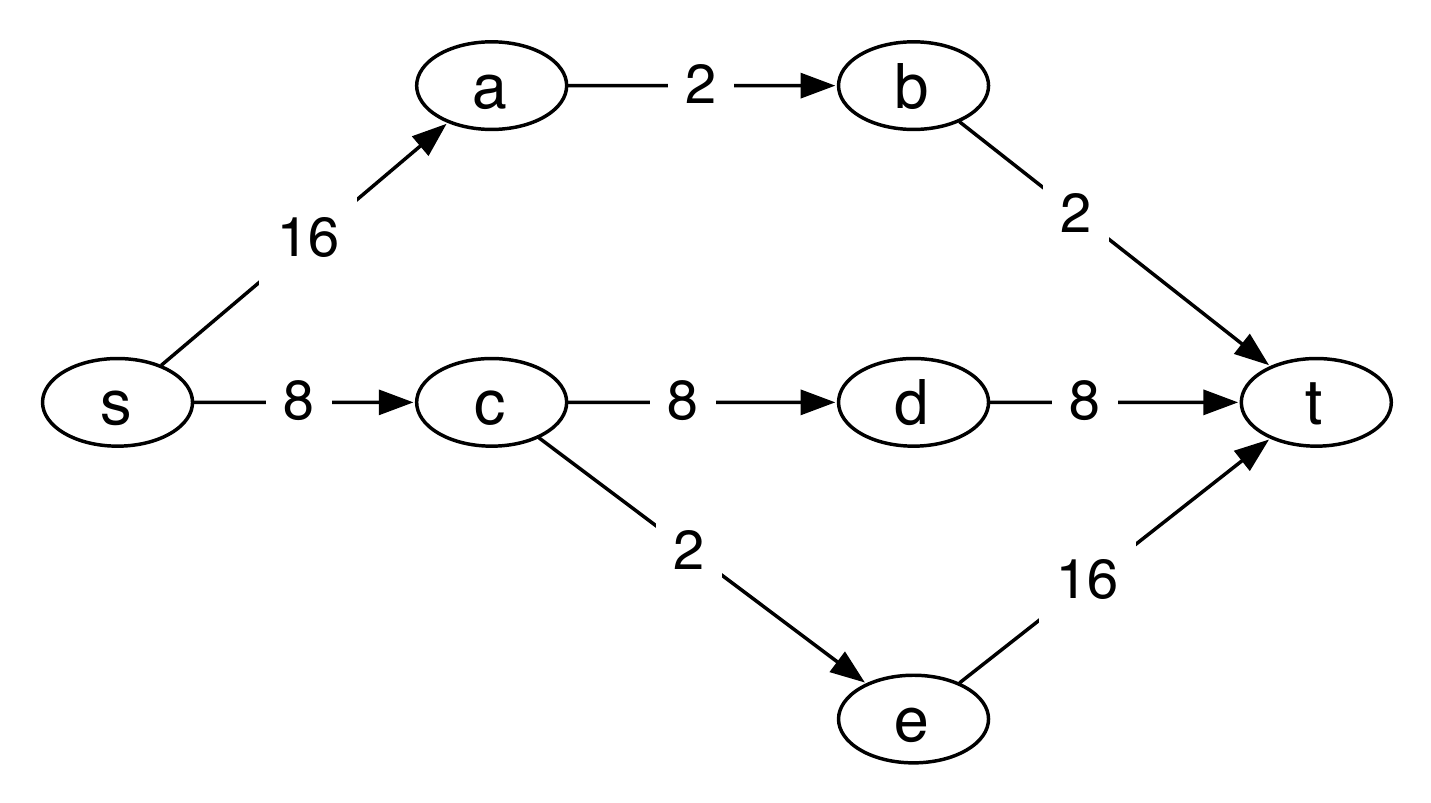}
\end{center}
\caption{A present-biased agent must choose a path from $s$ to $t$.}
\label{fig:intro-ex01}
\end{figure}

This example illustrates a few points.  First, when the
agent set out on the edge $(s,c)$, it was intending to next follow
the edge $(c,d)$, but when it got to $c$, it changed its mind and
followed the edge $(c,e)$.  A time-consistent agent (with $\beta = 1$), 
by contrast, would never do this; the path it decides to take starting at $s$
is the path it will continue to follow all the way to $t$.
Second, we are interested in whether the agent minimizes the cost
of traveling from $s$ to $t$ according to the real costs, not according
to its evaluation of the costs, and in this regard it fails to do so;
the shortest path is $s$-$a$-$b$-$t$, with a cost of $20$, while
the agent incurs a cost of $26$.

\xhdr{Overview of Results}
Our graph-theoretic framework makes it possible to reason about 
time-inconsistency effects that arise in very different settings,
provided simply that the underlying decisions faced by the agent can be
modeled as the search for a path through a graph-structured
sequence of options.
And perhaps more importantly, since it is tractable to ask questions
that quantify over all possible graphs, we can cleanly compare
different scenarios, and search for the best or worst possible
structures relative to specific objectives.
This is difficult to do without an underlying combinatorial structure.
For example, suppose we were inspired by Akerlof's example to try
identifying the scenario in which time-inconsistency leads to the
greatest waste of effort.
{\em A priori}, it is not clear how to formalize the search over
all possible ``scenarios.''
But as we will see, this is precisely something we can do if 
we simply ask for the graph in which time-inconsistency produces
the greatest ratio between the cost of the 
path traversed and cost of the optimal path.

Moreover, with this framework in place, it becomes easier to express
formal questions about design for these contexts: if as a designer of
a complex task we are able to specify the underlying graph structure,
which graphs will lead time-inconsistent agents to reach the goal
as efficiently as possible?

Our core questions are based on quantifying the inefficiency
from time-inconsistent behavior, 
designing task structures to reduce this inefficiency, and 
comparing the behavior of agents with different levels of time-inconsistency.
Specifically, we ask:
\begin{enumerate}
\item 
In which graph structures does time-inconsistent planning have the potential
to cause the greatest waste of effort relative to optimal planning?
\item 
How do agents with different levels of present bias (encoded as different
values of $\beta$) follow combinatorially different paths through 
a graph toward the same goal?
\item 
Can we increase an agent's efficiency in reaching a goal by deleting nodes
and/or edges from the underlying graph, thus reducing the number of options 
available?
\item 
How do we structure tasks for a heterogeneous collection of agents with
diverse values of $\beta$, so that each agent either reaches the goal
or abandons it quickly without wasting effort?
\end{enumerate}

In what follows, we address these questions in turn.
For the first question, we consider $n$-node graphs and 
ask how large the {\em cost ratio} can be between the path followed by
a present-biased agent and the path of minimum total cost.
Since deviations from the minimum-cost plan due to present bias
are sometimes viewed as a form of ``irrational'' behavior,
this cost ratio effectively serves as a ``price of irrationality''
for our system.
We give a characterization of the worst-case graphs in terms of
{\em graph minors}; this enables us to show, roughly speaking,
that any instance with sufficiently high cost ratio must contain
a large instance of the Akerlof example embedded inside it.

For the second question, we consider the possible paths followed
by agents with different present-bias parameters $\beta$.
As we sweep $\beta$ over the interval $[0,1]$, we have a type
of {\em parametric} path problem, where the choice of 
path is governed by a continuous parameter ($\beta$ in this case).
We show that in any instance, the number of distinct paths is bounded
by a polynomial function of $n$, which forms an interesting contrast
with canonical formulations of the parametric shortest-path problem,
in which the number of distinct paths can be superpolynomial in $n$
\cite{carstensen-parametric-sp,nikolova-parametric-sp}.

The third and fourth questions are essentially design questions:
we must design a set of tasks to optimize the performance
of a present-biased agent.
For the third question, we show how it is possible for agents to be more
efficient when nodes and/or edges are deleted from the underlying graph;
on the other hand, if we want to motivate an agent
to follow a particular path $P$ through the graph, it can be crucial
to present the agent with a subgraph that includes not just $P$ 
but also certain additional nodes and edges that do not belong to $P$.
We give a graph-theoretic characterization of the 
possible subgraphs supporting efficient traversal.
Finally, for heterogeneous agents, 
we explore a simple variant of the problem
based on partitioning large tasks into smaller ones.

Before turning to these questions, we first discuss the basic graph-theoretic
problem in more detail, showing how instances of this problem capture
the time-inconsistency phenomena discussed earlier in this section.

\section{The Graph-Theoretic Model}

In order to argue that our graph-theoretic model captures a variety
of phenomena that have been studied in connection with time-inconsistency,
we present a sequence of examples to illustrate some of the
different behaviors that the model exhibits.
We note that the example in Figure \ref{fig:intro-ex01} already 
illustrates two simple points: that the path chosen by the agent can
be sub-optimal; and that even if the agent traverses an edge $e$
with the intention of following a path $P$ that begins with $e$, 
it may end up following a different path $P'$ that also begins with $e$.

For an edge $e$ in $G$, let $c(e)$ denote the cost of $e$;
and for a path $P$ in $G$, let $e_i(P)$ denote the $i^{\rm th}$ edge on $P$.  
In terms of this notation, the agent's decision is easy to specify:
when standing at a node $v$, it chooses the path $P$ that minimizes
$c(e_1(P)) + \beta \sum_{i > 1} c(e_i(P))$ over all $P$ that run from $v$ 
to $t$.  It follows the first edge of $P$ to a new node $w$, and
then performs this computation again.

We begin by observing that Figure \ref{fig:akerlof-ex02}
represents a version of the Akerlof example from the introduction.
(For simplicity we assume that the delivery of the package is
instantaneous, so $h = 0$.  
Also recall that we use $b$ to denote $\beta^{-1}$.)
Node $t$ represents the state in which the agent has sent the package,
and node $v_i$ represents the state in which the agent has reached day $i$
without sending the package.
The agent has the option of going directly from node $s$ to node $t$,
and this is the shortest $s$-$t$ path.
But if $(b - 1) c > b x$, then the agent will instead go from $s$ to $v_1$,
intending to complete the path $s$-$v_1$-$t$ in the next time 
step. 
At $v_1$, however, the agent decides to go to $v_2$, intending 
to complete the path $v_1$-$v_2$-$t$ in the next time step.
This process continues: the agent, following exactly
the reasoning in the example from the introduction, is procrastinating
and not going to $t$, and in the end its path goes all the way to the
last node $v_n$ ($n = 5$ in the figure) 
before finally taking an edge to $t$.
(One minor change from the set-up in the introduction is the fact
that the present-bias effect here holds more consistently, and
is applied to $x$ as well; this has no real effect on the underlying story.)

\begin{figure}[t]
\begin{center}
\subfigure[\emph{The Akerlof example}]{
\includegraphics[width=2.50in]{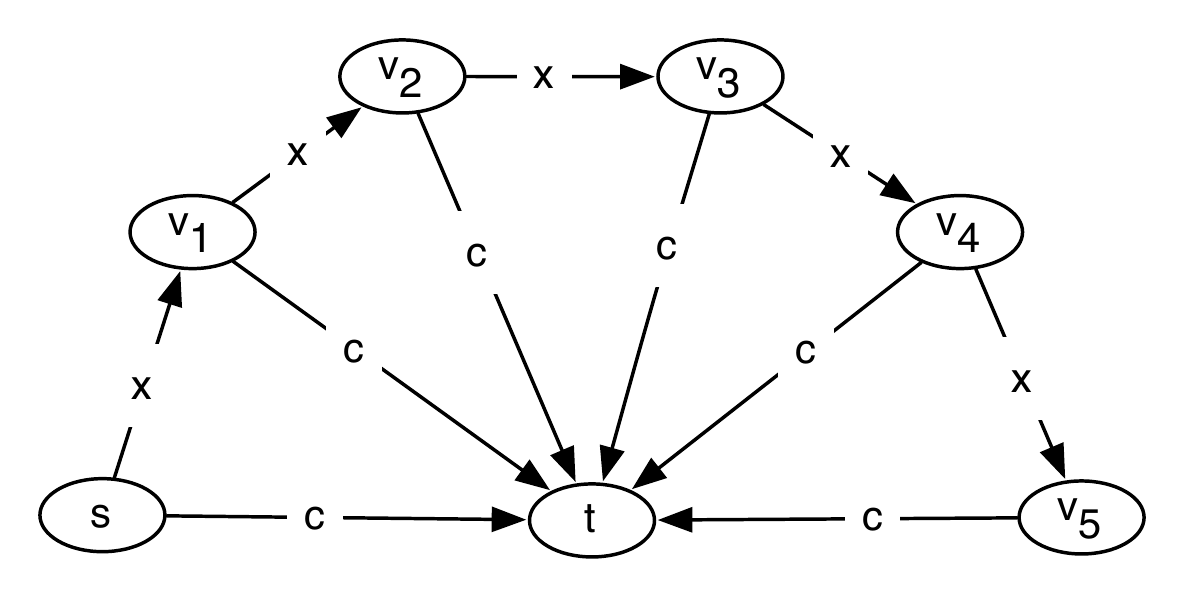}
 \label{fig:akerlof-ex02}
}
\hspace*{0.25in}
\subfigure[\emph{Homework deadlines}]{
\includegraphics[width=2.50in]{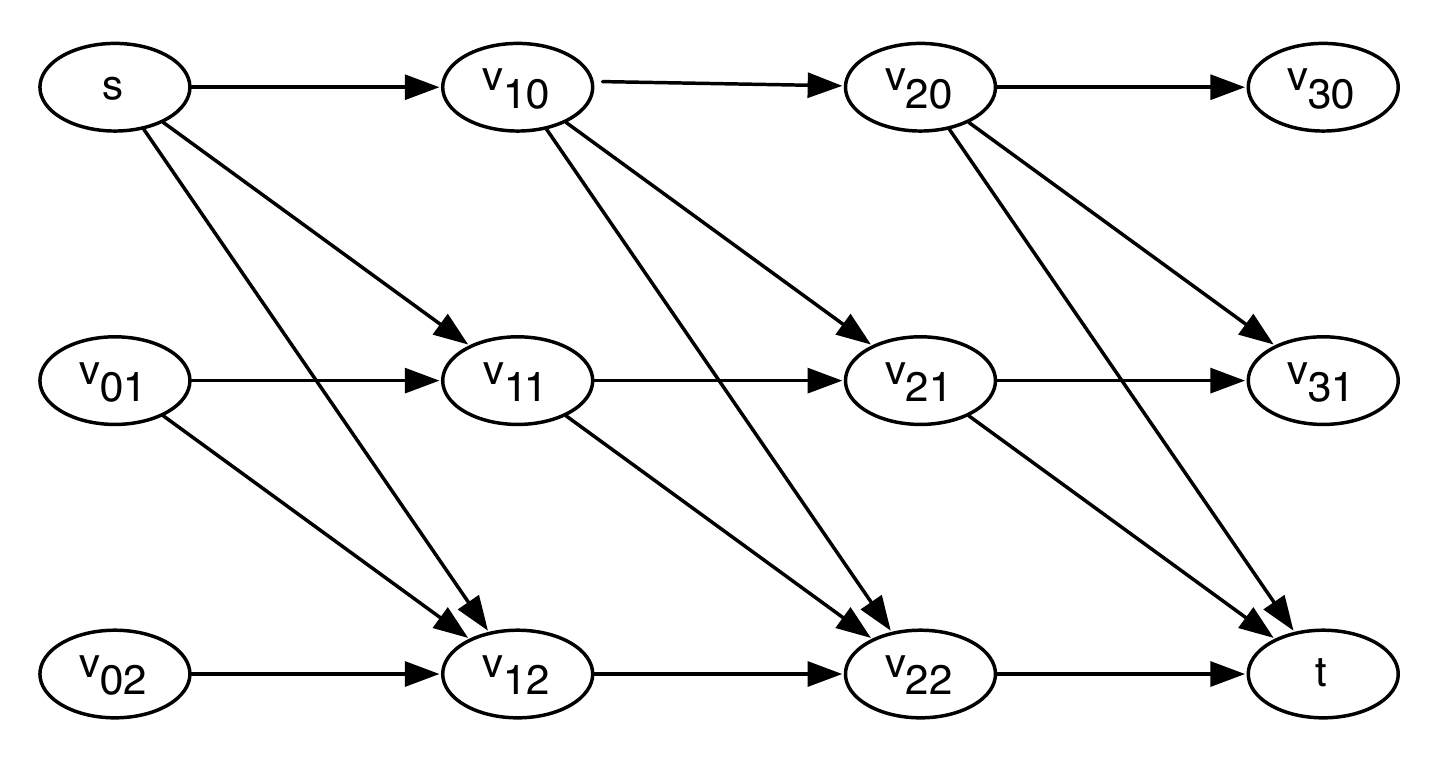}
 \label{fig:deadline-ex01}
}
\end{center}
\caption{Path problems that exhibit procrastination, abandonment, and
choice reduction.}
\label{fig:graph-ex}
\end{figure}

\xhdr{Extending the model to include rewards}
Thus far we can't talk about an agent who abandons its pursuit of
the goal midway through, since our model requires the agent to construct
a path that goes all the way to $t$.
But a simple extension of the model makes it possible to consider
such situations.

Suppose we place a reward of $r$ at the target node $t$, which will 
be claimed if the agent reaches $t$.
Standing at a node $v$, the agent now has an expanded set of options:
it can follow an edge out of $v$ as before, or it can quit taking steps,
incurring no further cost but also not claiming the reward.
The agent will choose the latter option precisely when either there is 
no $v$-$t$ path, or when the minimum 
cost of a $v$-$t$ path exceeds the value of the reward, evaluated
in light of present bias: 
$c(e_1(P)) + \beta \sum_{i > 1} c(e_i(P)) > \beta r$ for
all $v$-$t$ paths $P$.
It is important to note a key feature of this evaluation: the
reward is always discounted by $\beta$ relative to the cost that
is being incurred in the current period, even if the reward will
be received right after this cost is incurred.
(For example, if the path $P$ has a single edge, then the agent
is comparing $c(e_1(P))$ to $\beta r$.)

In what follows, we will consider both these models: the former
{\em fixed-goal model}, in which the agent must reach $t$ and
seeks to minimize its cost; and the latter {\em reward model} 
in which the agent trades off cost incurred against reward at $t$,
and has the option of stopping partway to $t$.
Aside from this distinction, both models share the remaining ingredients,
based on traversing an $s$-$t$ path in $G$.

It is easy to see that the reward model displays the phenomenon of
{\em abandonment}, in which the agent spends some cost to try reaching
$t$, but then subsequently gives up without receiving the reward.
Consider for example a three-node path on nodes $s$, $v_1$, and $t$,
with an edge $(s,v_1)$ of cost $1$ and an edge $(v_1,t)$ of cost $4$.
If $\beta = 1/2$ and there is a reward of $7$ at $t$, then the agent
will traverse the edge $(s,v_1)$ because it evaluates the total 
cost of the path at $1 + 4 \beta = 3 < 7 \beta = 3.5$.
But once it reaches $v_1$, it evaluates the cost of completing
the path at $4 > 7 \beta = 3.5$, and so it quits without reaching $t$.

\xhdr{An example involving choice reduction}
It is useful to describe a more complex example that shows the
modeling power of this shortest-path formalism, and also shows how we can
use the model to analyze deadlines as a form of beneficial choice reduction.
(As should be clear, with a time-consistent agent it can never help 
to reduce the set of choices; such a phenomenon requires some form
of time-inconsistency.)
First we describe the example in text, and then show how to represent
it as a graph.

Imagine a student taking a three-week short course in which the 
required work is to complete two small projects by the end of the course.
It is up to the student when to do the projects, as long as they
are done by the end.
The student incurs an effort cost of $1$ from any week in which she
does no projects (since even without projects there is still the
lower-level effort of attending class),
a cost of $4$ from any week in which she does one
project, and a cost of $9$ from any week in which she does both projects.
Finally, the student receives a reward of $16$ for completing 
the course, and she has a present-bias parameter of $\beta = 1/2$.

Figure \ref{fig:deadline-ex01} shows how to represent this scenario
using a graph.  Node $v_{ij}$ corresponds to a state in which $i$ weeks
of the course are finished, and the student has completed $j$ projects 
so far; we have $s = v_{00}$ and $t = v_{32}$.
All edges go one column to the right, indicating that one week will
elapse regardless; what is under the student's control is how many
rows the edge will span.
Horizontal edges have cost $1$, edges that descend one row have cost $4$, 
and edges that descend two rows in a single hop have cost $9$.
In this way, the graph precisely represents the story just described.

How does the student's construction of an $s$-$t$ path work out?
From $s$, she goes to $v_{10}$ and then to $v_{20}$, intending
to complete the path to $t$ via the edge $(v_{20}, t)$.
But at $v_{20}$, she evaluates the cost of the edge $(v_{20},t)$ as
$9 > \beta r = 16/2 = 8$, and so she quits without reaching $t$.
The story is thus a familiar one: the student plans to do both projects
in the final week of the course, but when she reaches the final week,
she concludes that it would be too costly and so she drops the course instead.

The instructor can prevent this from happening through a very simple
intervention.  If he requires that the first project be completed
by the end of the second week of the course, this corresponds simply to 
deleting node $v_{20}$ from the graph.
With $v_{20}$ gone, the path-finding problem changes: now the student
starting at $s$ decides to follow the path $s$-$v_{10}$-$v_{21}$-$t$,
and at $v_{10}$ and then $v_{21}$ she continues to select this path,
thereby reaching $t$.
Thus, by reducing the set of options available to the student --- and
in particular, by imposing an intermediate deadline to enforce progress ---
the instructor is able to induce the student to complete the course.

There are many stories like this one about homework and deadlines,
and our point is not to focus too closely on it in particular.
Indeed, to return to one of the underpinnings of our graph-theoretic formalism,
our point is in a sense the opposite:
it is hard to reason about the space of possible ``stories,''
whereas it is much more tractable to think about the space of possible graphs.
Thus by encoding the set of stories mechanically in the form of graphs,
it becomes feasible to reason about them as a whole.

We have thus seen how a number of different time-inconsistency 
phenomena arise in simple instances of the path-finding problem.
The full power of the model, however, lies in 
proving
statements that quantify over all graphs;
we begin this next.

\def\cti{{c_\beta}}
\def\pti{{P_\beta}}
\def\cb{{\bf{c}}}
\def\poa{{\theta}}
\def\tent{{\cal T}}
\def\F{{\cal F}}
\def\orient{{\zeta}}
\def\skel{{\sigma}}
\def\ch{{\cal P}}

\section{The Cost Ratio: A Characterization Via Graph Minors}

Our path-finding model naturally motivates a 
basic quantity of interest: the {\em cost ratio},
defined as the ratio between the cost of the
path found by the agent and the cost of the shortest path.
We work here within the fixed-goal version of the model,
in which the agent is required to reach the goal $t$ and
the objective is to minimize the cost of the path used.

To fix notation for this discussion, given 
a directed acyclic graph $G$ on $n$ nodes
with positive edge costs, we let $d(v,w)$ denote the cost of
the shortest $v$-$w$ path in $G$, when one exists
(using the true edge costs, not modified by present bias).
We are given designated
start and end nodes $s$ and $t$ respectively, 
and we pre-process the graph by deleting every node that
is either not reachable from $s$ or that cannot reach $t$.
Let $\pti(v,t)$ denote the the $v$-$t$ path followed by
an agent with present-bias $\beta$, and
let $\cti(v,t)$ be the total cost of this path.
The cost ratio can thus be written as $\cti(s,t)/d(s,t)$.

\xhdr{A bad example for the cost ratio}
We first describe a simple construction showing that the cost ratio
can be exponential in the number of nodes $n$.
We then move on to the main result of this section, which is a
characterization of the instances in which the cost ratio achieves
this exponential lower bound.

Our construction is an adaptation of the Akerlof example from 
the introduction.
We describe it using edges of zero cost, but it is easy to modify
it to give all edges positive cost.
We have a graph that consists of a directed path $s = v_0, v_1, v_2, 
\ldots, v_n$, and with each $v_i$ also linking directly to node $t$.
(The case $n = 5$ is the graph in Figure \ref{fig:akerlof-ex02}.)
With $b = \beta^{-1}$, we choose any $\mu < b$;
we let the cost of the edge $(v_j,t)$ be $\mu^j$,
and let the cost of each edge $(v_j, v_{j+1})$ be $0$.

Now, when the agent is standing at node $v_j$, it evaluates the
cost of going directly to $t$ as $\mu^j$, while the cost of 
the two-step path through $v_{j+1}$ to $t$ is evaluated as
$0 + \beta \mu^{j+1} = (\beta \mu) \mu^j < \mu^j$.
Thus the agent will follow the edge $(v_j,v_{j+1})$ with
the plan of continuing from $v_{j+1}$ to $t$.
But this holds for all $j$, so once 
it reaches $v_{j+1}$, it changes its mind and continues on to $v_{j+2}$,
and so forth.
Ultimately it reaches $v_n$, and then must go directly to $t$ at
a cost of $\cti(s,t) = \mu^n$.
Since $d(s,t) = 1$ by using the edge directly from $s$ to $t$,
this establishes the exponential lower bound on the cost ratio
$\cti(s,t)/d(s,t)$.
Essentially, this construction shows that the Akerlof example can be
made quantitatively much worse than its original formulation
by having the cost of going directly to the goal
grow by a modest constant factor in each time step; 
when a present-biased agent procrastinates in this case,
it ultimately incurs an exponentially large cost.

As noted above, the fact that some edges in this example have zero cost
is not crucial; we could give each edge $(v_j, v_{j+1})$ a uniform cost
$\eps > 0$ and correspondingly reduce the value of $\mu$ slightly.

\xhdr{A Graph Minor Characterization}
We now provide a structural description of the graphs on which
the cost ratio can be exponential in the number of nodes $n$ ---
essentially we show that a constant fraction of the nodes in 
such a graph must have the structure of the Akerlof example.

We make this precise using the notion of a {\em minor}.
Given two undirected graphs $H$ and $K$, we say that $H$ {\em contains
a $K$-minor} if we can map each node $\kappa$ of $K$ to a connected
subgraph $S_\kappa$ in $H$, with the properties that 
(i) $S_\kappa$ and $S_{\kappa'}$ are disjoint for every two nodes
$\kappa, \kappa'$ of $K$, and (ii) if $(\kappa, \kappa')$ is an edge of $K$,
then in $H$ there is some edge connecting a node in $S_\kappa$ to
a node in $S_{\kappa'}$.  Informally, the definition means that
we can build a copy of $K$ using the structure of $H$, with 
disjoint connected subgraphs of $H$ playing the role of ``super-nodes''
that represent the nodes of $K$, and with the adjacencies among
these super-nodes representing the adjacencies in $K$.
The minor relation shows up in many well-known results in graph
theory, perhaps most notably in Kuratowski's Theorem that a non-planar
graph must contain either the complete graph $K_5$ or the complete
bipartite graph $K_{3,3}$ as a minor
\cite{diestel-graph-theory-book}.

Our goal here is to show that if $G$ has exponential cost ratio, then
its undirected version must contain a large copy of the graph
underlying the Akerlof example as a minor.
In other words, the Akerlof example is not only a way to produce a
large cost ratio, but it is in a sense an unavoidable signature
of any example in which the cost ratio is very large.

We set this up as follows.
Let $\skel(G)$ denote the skeleton of $G$,
the undirected graph obtained by removing
the directions on the edges of $G$.
Let $\F_k$ denote the graph with nodes 
$v_1, v_2, \ldots, v_k$, and $w$, and edges 
$(v_i,v_{i+1})$ for $i = 1, \ldots, k-1$, and
$(v_i, w)$ for $i = 1, \ldots, k$.
We refer to $\F_k$ as the {\em $k$-fan}.

We now claim

\begin{theorem}
For every $\lambda > 1$ there exist $n_0 > 0$ and $\eps > 0$ such that if
$n \geq n_0$ and $\cti(s,t) / d(s,t) > \lambda^n$, then $\skel(G)$ contains
an $\F_k$-minor for some $k \geq \eps n$.
\label{thm:minor}
\end{theorem}

\begin{proof}
The idea behind the proof is as follows.
For each node, we consider a ``rounded'' version of its distance to $t$ ---
essentially the logarithm of its distance to $t$, truncated
down to an integer value.
We argue that there exists a node $v$ on the path taken by the agent
for which this quantity is large, and we consider the portion $P$ of this path
that runs from $s$ to $v$.
From many nodes on $P$, other paths emanate back to $t$ that
remain disjoint from $P$.  These other paths together with $P$
provide us with the structure from which we can build an $\F_k$-minor.

We now give the detailed argument.
For each node $v$, we define the {\em rank} of $v$, denoted $r(v)$, to be
$0$ if $d(v,t) \leq d(s,t)$, and otherwise it is
the minimum integer $j > 0$ such that $d(v,t) \leq b^j d(s,t)$.

Here is a first basic fact about ranks.
\begin{quote}
{\em 
(A)
If $(v,w)$ is an edge on $\pti(s,t)$, then $r(w) \leq r(v) + 1$.
}
\end{quote}

To prove (A), consider an edge $(v,w)$ on $\pti(s,t)$.  If $(v,w)$ lies 
on a shortest path from $v$ to $t$, then $d(w,t) \leq d(v,t)$ and
hence $r(w) \leq r(v)$.
Otherwise, let $(v,w')$ be an edge on a shortest $v$-$t$ path;
the agent's decision to traverse the edge $(v,w)$ means that 
$$c(v,w) + \beta d(w,t) \leq c(v,w') + \beta d(w',t) \leq d(v,t),$$
and hence $d(w,t) \leq b d(v,t)$.
It follows that $r(w) \leq r(v) + 1$.

Now, suppose that $\cti(s,t) > \lambda^n d(s,t)$.
For future use, we choose constants
$n_0 > 0$ and $\lambda_0 > 1$ such that
when $n \geq n_0$, we have
$n \lambda_0^n < \lambda^n$.
We then set $\eps = \log_b \lambda_0$.

We now claim
\begin{quote}
{\em 
(B)
There exists an edge $(v,w)$ on $\pti(s,t)$ of cost $> \lambda_0^n d(s,t)$.
}
\end{quote}

To prove (B), we observe that 
since the quantity $\cti(s,t)$ is a sum of at most $n$ terms, 
corresponding to the edge costs in $\pti(s,t)$, there exists at
least one term in this sum that is $> \lambda^n d(s,t) / n$.
Suppose it is the cost of the edge $(v,w)$ on $\pti(s,t)$; i.e.
$c(v,w) > \lambda^n d(s,t) / n$.
By our choice of $n_0$ and $\lambda_0$, we then have 
$c(v,w) > \lambda_0^n d(s,t)$ when $n \geq n_0$.

\begin{figure}[t]
\begin{center}
\includegraphics[width=3.00in]{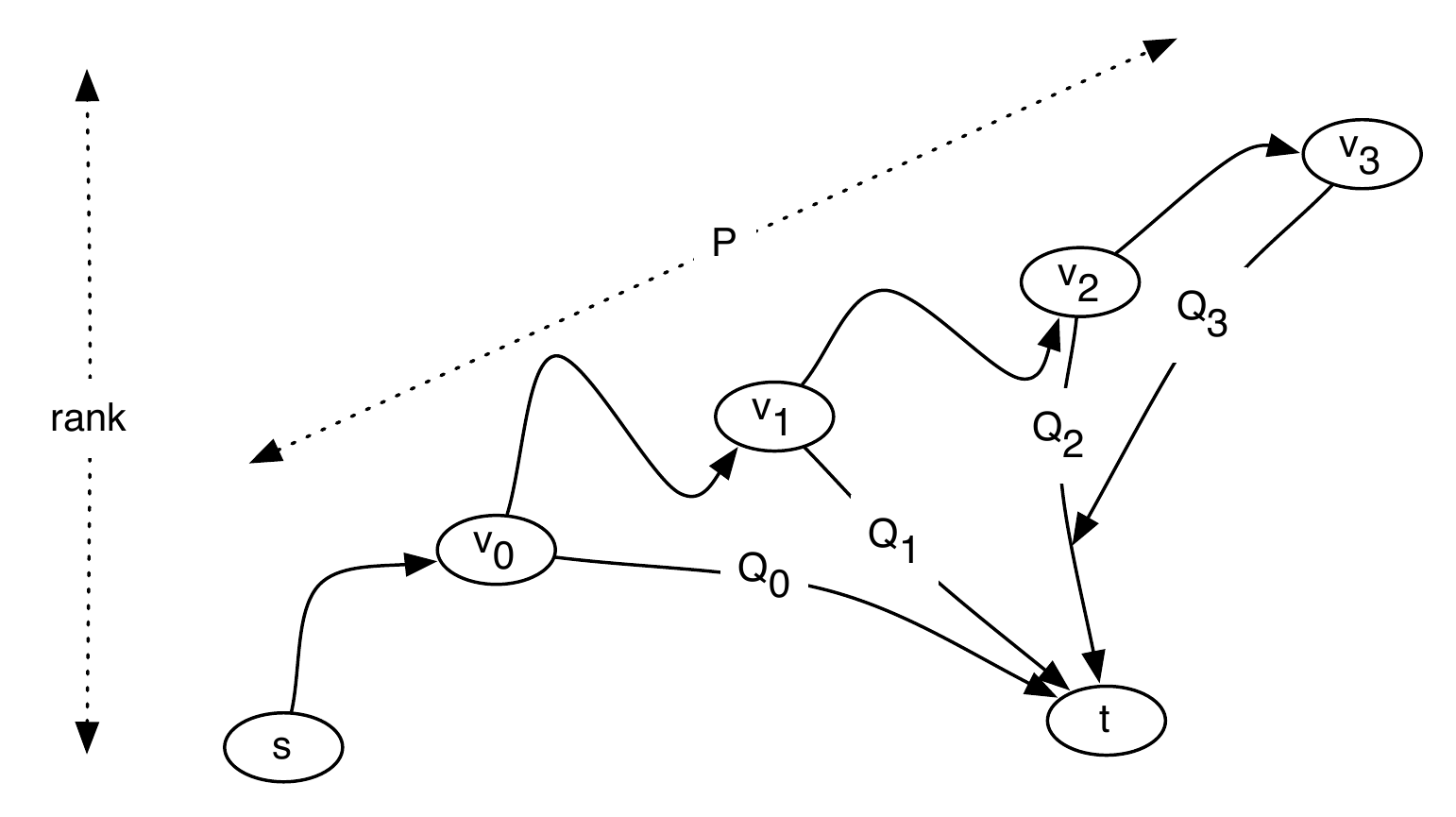}
\end{center}
\caption{The construction of an $\F_k$-minor in 
Theorem \ref{thm:minor}.
\label{fig:minor-fig01}
}
\end{figure}

Next we claim
\begin{quote}
{\em 
(C)
There exists a node $v$ on $\pti(s,t)$ of rank at least $\eps n$.
}
\end{quote}

Indeed, the tail $v$ of the edge $(v,w)$ from (B)
has rank at least $\eps n$.
To see this, consider the agent's decision at node $v$, when it chooses to pay
$c(v,w) > \lambda_0^n d(s,t)$ on its next step.
If $(v,w')$ is the first edge on a shortest $v$-$t$ path, then 
we have 
$d(v,t) = c(v,w') + d(w',t) \geq 
c(v,w') + \beta d(w',t) \geq c(v,w) + \beta d(w,t) \geq c(v,w) 
> \lambda_0^n d(s,t) = b^{\eps n} d(s,t)$,
where $\eps = \log_b \lambda_0$.

The ingredients for the remainder of the proof are depicted
schematically in Figure \ref{fig:minor-fig01}.
Let $P$ denote the subpath of $\pti(s,t)$ consisting 
of the portion from $s$ to $v$, where $v$ is the node defined in 
the statement of (C).
By (A) we know that ranks of consecutive nodes on $P$ can differ
by at most $1$, and so there are nodes on $P$ of each rank from $0$ to $k$,
where $k = \lceil \eps n \rceil$.
For each $j$ from $0$ to $k$, let $v_j$ denote 
the last node of rank $j$ on $P$.

Let $Q_j$ be a shortest $v_j$-$t$ path.
We claim 
\begin{quote}
{\em 
(D)
$Q_j \cap P = \{v_j\}$.  
}
\end{quote}

Indeed, since $d(w,t) \leq d(v_j,t)$
for every node $w \in Q_j$, it follows that 
every node on $Q_j$ has rank at most $j$.
By the definition of $v_j$, all nodes on $P$ that come after $v_j$
have rank greater than $j$, and so no node on $P$ after $v_j$ can
belong to $Q_j$.
And if a node $w$ on $P$ before $v_j$ belonged to $Q_j$, then $G$
would contain a cycle, by combining the path from $w$ to $v_j$ using $P$,
followed by the path from $v_j$ to $w$ using $Q_j$.
Thus $v_j$ is the only node in $Q_j \cap P$, and
this proves (D).

We are now prepared to construct an $\F_{k+1}$-minor in $\skel(G)$.
First, we partition the path $P$ into disjoint segments 
$R_0, R_1, \ldots, R_k$ such that $v_j \in R_j$.
Now we observe that $S = \cup_{j=0}^k (Q_j - \{v_j\})$ is a connected
set of nodes, since all $Q_j$ contain $t$.
By (D), the set $S$ is also disjoint from each $R_j$.
Since there is an edge between $R_j$ and $R_{j+1}$ for each $j$,
and also an edge from $S$ to $v_j \in R_j$ for each $j$,
it follows that $R_0, R_1, \ldots, R_k, S$ form the super-nodes
in an $\F_{k+1}$-minor.
This completes the proof of Theorem \ref{thm:minor}.
\end{proof}

\section{Collections of Heterogeneous Agents}
\label{sec:path-counting}

\def\P{{\cal P}}

Thus far we have focused on the behavior of a single agent with
a given present-bias parameter $\beta$.
Now we consider all possible values of $\beta$, and ask the following
basic question: how large can the set $\{\pti(s,t) : \beta \in [0,1]\}$ be?
In other words, if for each $\beta$, an agent with parameter $\beta$
were to construct an $s$-$t$ path in $G$, how many different paths
would be constructed across all the agents?
Bounding this quantity tells us how many genuinely
``distinct'' types of behaviors there are for the instance defined by $G$.

Let $\P(G)$ denote the set $\{\pti(s,t) : \beta \in [0,1]\}$.
Despite the fact that $\beta$ comes from the continuum $[0,1]$,
the set $\P(G)$ is clearly finite, 
since $G$ only has finitely many $s$-$t$ paths.
The question is whether we can obtain a non-trivial upper bound
on the size of $\P(G)$, and in particular
one that does not grow exponentially in the number of nodes $n$.

In fact this is possible, and our main goal in this section is
to prove the following theorem.

\begin{theorem}
For every directed acyclic graph $G$, the size of 
$\P(G)$ is $O(n^2)$.
\label{thm:path-counting}
\end{theorem}

As noted in the introduction, this can be viewed as a novel
form of parametric path problem, and for the more standard
genres of parametric path problems, the number of possible
paths has a superpolynomial lower bound. (Probably the most typical
example is to track the set of all shortest paths in a directed
graph when each edge has a cost of the form $a_e x + b_e$ and
$x$ ranges over $[0,1]$
\cite{carstensen-parametric-sp,nikolova-parametric-sp}.)
For our path-counting problem, we need a technique that
gives a polynomial upper bound.

We prove Theorem \ref{thm:path-counting} by first defining a more general 
path-counting problem, and then showing that bounding $\P(G)$
is an instance of this problem.
We begin with two technical considerations.
First, 
as at earlier points in the paper, we assume that we have
pre-processed $G$ so that every node and edge lies on some $s$-$t$
path in $G$; any node or edge that doesn't have this property
is not relevant to counting $s$-$t$ paths and can be deleted
without affecting the result.
Second, we assume that when an agent with present-bias parameter $\beta$
is indifferent between two edges leaving a node $v$ --- i.e. they both
evaluate to equal cost --- then it uses a consistent tie-breaking rule,
such as choosing to go to the node that is earlier in a fixed
topological ordering of $G$.

\xhdr{The Interval Labels Problem}
The more general problem we study is something we call the 
{\em Interval Labels Problem}, and it is defined as follows.
In the Interval Labels Problem, we are given
a directed acyclic graph $G$ with a distinguished source node $s$
and target node $t$, such that every node and edge of $G$ lies on
some $s$-$t$ path.
For each node $v$ of $G$, let $\delta^{out}(v)$ denote the edges
emanating from $v$;
each $e \in \delta^{out}(v)$ is assigned an interval $I(e) \subseteq [0,1]$
such that the intervals $\{I(e) : e \in \delta^{out}(v)\}$ partition $[0,1]$.
This defines an instance of the Interval Labels Problem.

We say that an $s$-$t$ path $P$ in $G$ is {\em valid} if
the intersection $\cap_{e \in P} I(e)$ is non-empty.
We say that $x \in [0,1]$ is a {\em witness} for the
$s$-$t$ path $P$ if $x \in \cap_{e \in P} I(e)$.
Thus a path is valid if and only if it has at least one witness.
The goal in the Interval Labels Problem is to count the number
of valid paths in $G$.

We first justify why counting the number of paths in 
$\P(G) = \{\pti(s,t) : \beta \in [0,1]\}$
can be reduced to an instance of the Interval Labels Problem.

\begin{lemma}
Given a directed acyclic graph $G$ with edge costs,
it is possible to create an instance of the Interval Labels Problem
for which the number of valid paths is precisely the size of $\P(G)$.
\label{lemma:interval-reduction}
\end{lemma}

\begin{proof}
For a given node $v$, let $w_1, w_2, \ldots, w_h$ be its out-neighbors;
each $w_j$ has an edge cost $c(v,w_j)$ from $v$ and a
shortest-path distance $d(w_j,t)$ to $t$.
For which values of $\beta$ will an agent standing at $v$ choose to 
go to $w_j$?  It must be the case that
$c(v,w_j) + \beta d(w_j,t) \leq c(v,w_i) + \beta d(w_i,t)$ for all 
$i \neq j$.
Thus, if we define $L_j(\beta)$ to be the line 
$c(v,w_j) + \beta d(w_j,t)$ as a function of $\beta$ defined over
the interval $[0,1]$, the values of $\beta$ at which the edge $(v,w_j)$
is chosen are those $\beta$ for which $L_j$ lies on the lower envelope
of the line arrangement $\{L_1, L_2, \ldots, L_h\}$.
We know from the structure of line arrangements that this set of $\beta$
is an interval $I_j$, and that the intervals $I_1, I_2, \ldots, I_h$
(some of which may be empty) partition $[0,1]$
\cite{edelsbrunner-comp-geom-book}.

We define an instance of the Interval Labels Problem by setting
the interval $I(e)$ for an edge $(v,w_j)$ to be this interval $I_j$.
Now, if $P$ is a path such that $\beta$ belongs to $I(e)$ for each $e \in P$,
then an agent with parameter $\beta$ would select each edge of $P$ in
sequence, and so $P \in \P(G)$.
Conversely, if $P \in \P(G)$, then $\beta$ is a witness for the path $P$,
and so $P$ is valid in the instance of the Interval Labels Problem.
As a result, the number of valid paths in this instance of 
the Interval Labels Problem is equal to the size of $\P(G)$.
This completes the proof of Lemma \ref{lemma:interval-reduction}
\end{proof}

It is therefore enough to put an upper bound on the number of
valid paths in any instance of the Interval Labels Problem, and
we do that in the following lemma.

\begin{lemma}
For any instance of the Interval Labels Problem, the number of
valid paths is at most the number of edges of $G$.
\label{lemma:valid-paths}
\end{lemma}
\begin{proof}
Note that since the intervals on the edges emanating from any node partition
$[0,1]$, a number $x$ can be 
a witness for just a single path, which we denote $P_x$.
Conversely, each valid path has at least one witness in $[0,1]$.

To count valid paths, we use the following scheme
that {\em charges} paths uniquely to edges of $G$ that they contain.
We begin with $x = 0$, identify the path
$P_0$ for which it is a witness, and then begin increasing $x$ continuously
until we reach the minimum $x^*$ at which $P_{x^*} \neq P_0$.
Let $e = (v,w)$ be the first edge on $P_0$ that is not in $P_{x^*}$. 
For this edge, we have $x^* \not\in I(e)$; it follows
that $I_e$ is a proper subset of $[0,x^*]$, and hence that
$e$ cannot belong to $P_x$ for any $x \geq x^*$.
We charge $P_0$ to $e$, and since we will maintain the property
that every path is charged to an edge it contains, we will not
charge any further paths to $e$.

In general, each time the path $P_y$ changes, at a value $y^*$,
to a new path $P_{y^*}$, the first edge $f$ of $P_y$ that is not on $P_{y^*}$
must have the property that the right endpoint of $I(f)$ is strictly
below $y^*$, and so the edge $f$ will never appear again on a path
in our counting process.
We charge $P_y$ to $f$.
No path has thus far been charged to $f$, since this is the first
time when a path has had a witness that lies to the right of $I(f)$,
and since $f$ will not appear on any future paths, no path will
be charged to it again.
Thus, $P_y$ is the only path that gets charged to $f$.

Continuing in this way, each path in our counting process gets associated
with a distinct edge of $G$, and hence 
the total number of valid paths
must be at most the number of edges of $G$.
\end{proof}

By Lemma \ref{lemma:interval-reduction}, 
the number of paths in $\P(G)$ is equal to the number of valid paths
in the equivalent instance of the Interval Labels Problem, 
which by Lemma \ref{lemma:valid-paths} is at most the number of
edges of $G$.
Since $G$ has $O(n^2)$ edges, Theorem \ref{thm:path-counting} follows.

\xhdr{A Tight Lower Bound}
We now show that the quadratic bound on the size of $\P(G)$ can't
be improved in the worst case.  We do this by first establishing 
a quadratic lower bound construction for the Interval Labels Problem.

\begin{proposition} \label{prop:interval-example}
There exists an instance of the 
Interval Labels Problem on an $n$-node graph for which 
the number of valid paths is $\Omega(n^2)$.
\end{proposition}
\begin{proof}
We consider a complete bipartite graph 
where the nodes on the left are $u_1, \ldots,
u_n$, the nodes on the right are $v_1, \ldots, v_n$, and there
is an edge $(u_i,v_j)$ for all $1 \leq i, j \leq n$.
There is a source node $s$ with edges $(s,u_i)$ for all $i$,
and a target node $t$ with edges $(v_j,t)$ for all $j$.

Now, the interval on
$(s,u_i)$ is $[\frac{i-1}{n}, \frac{i}{n}]$. The interval on $(v_j,t)$
is $[0,1]$, as there is only one edge out of $v_j$. The interval on
$(u_i,v_j)$ is $[\frac{i-1}{n} + \frac{j-1}{n^2}, \frac{i-1}{n} +
\frac{j}{n^2}]$, except that we extend the interval for $j = 1$ down
to $0$ and the interval for $j = n$ up to $1$ so that these intervals
form a partition of $[0,1]$.  Under this construction for any $i$ and
$j$ the path $s,u_i,v_j,t$ is a valid path since the intersection of
the three edges on this path is simply the interval of the middle edge,
$[\frac{i-1}{n} + \frac{j-1}{n^2}, \frac{i-1}{n} + \frac{j}{n^2}]$.
Therefore, this instance admits a quadratic number of valid paths.
\end{proof}

We can show how to add edge costs to the example in 
Proposition \ref{prop:interval-example} so that all the valid paths
become members of $\P(G)$.  We omit the details of the construction
in this version.

\begin{proposition} \label{prop:pg-example}
There exists a directed acyclic graph $G$ with edge costs for
which $\P(G)$ has size $\Omega(n^2)$.
\end{proposition}

\section{Motivating an Agent to Reach the Goal}

We now consider the version of the model with rewards: 
there is a reward at $t$, and the agent has the additional option
of quitting if it perceives --- under its present-biased evaluation ---
that the value of the reward is not worth the remaining cost in the path.

Note that the presence of the reward does not affect the agent's
{\em choice} of path, only whether it continues along the path.
Thus we can clearly determine the minimum reward $r$ required to motivate
the agent to reach the goal in $G$ by simply having it construct
a path to $t$ according to our standard fixed-goal model, 
identifying the node at which it perceives the remaining cost to be
the greatest (due to present bias this might not be $s$), and
assigning this maximum perceived cost as a reward at $t$.

A more challenging question is suggested by the possibility 
of deleting nodes and edges from $G$; recall that
Figure \ref{fig:deadline-ex01} showed a basic example in 
which the instructor of a course was able to motivate a student
to finish the coursework by deleting a node from the underlying graph.
(This deletion essentially corresponded to introducing a deadline
for the first piece of work.)
This shows that even if the reward remains fixed, in general 
it may be possible for a designer
to remove parts of the graph, thereby reducing the set of
options available to the agent, so as to get the agent to reach the goal.
We now consider the structure of the subgraphs that naturally
arise from this process.

\begin{figure}[t]
\begin{center}
\includegraphics[width=2.50in]{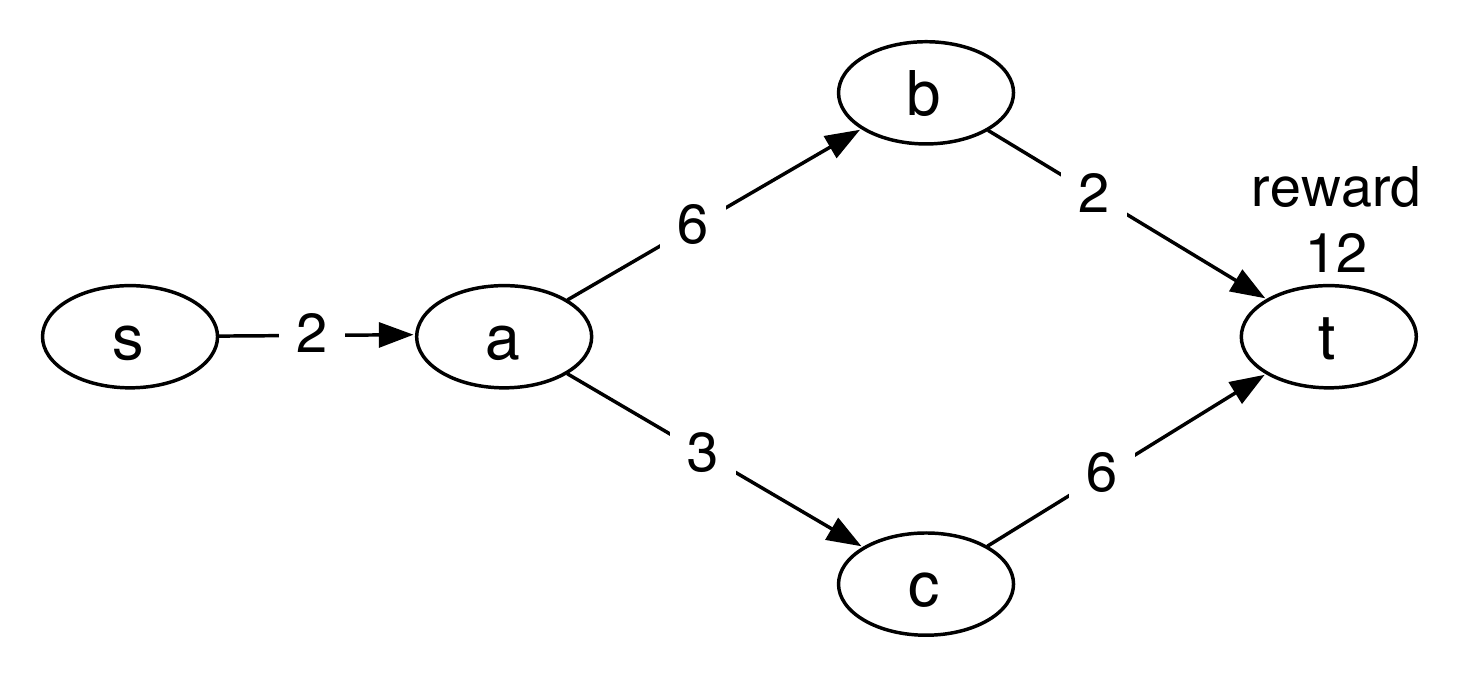}
\end{center}
\caption{A minimal subgraph for motivating an agent to reach $t$.
\label{fig:minimal-ex01}
}
\end{figure}

\xhdr{Motivating subgraphs: A fundamental example}
The basic set-up we consider is the following.
Suppose the agent in the reward model is trying to construct a 
path from $s$ to $t$ in $G$; the reward $r$ is not under our control ---
perhaps it is defined by a third party, or represents an
intrinsic reward that we cannot augment --- but we are able to
remove nodes and edges from the graph (essentially by declaring certain 
activities invalid, as the deadline did in Figure \ref{fig:deadline-ex01}).
Let us say that a subgraph $G'$ of $G$ {\em motivates} the agent if
in $G'$ with reward $r$, the agent reaches the goal node $t$.
(We will also refer to $G'$ as a {\em motivating subgraph}.)
Note that it is possible for the full graph $G$ to be a motivating
subgraph (of itself).

It would be natural to conjecture that if there is any subgraph $G'$
of $G$ that motivates the agent, then 
there is a motivating subgraph consisting simply of an $s$-$t$ path $P$.
Indeed, in any motivating subgraph $G'$,
the actual sequence of nodes and edges
the agent traverses does form an $s$-$t$ path $P$, 
and so one might suspect that
this path $P$ on its own should also be motivating.

In fact this is not the case, however.
Figure \ref{fig:minimal-ex01} shows a graph illustrating a phenomenon
that we find somewhat surprising {\em a priori}, though not hard to verify
from the example.
In the graph $G$ depicted in the figure, an agent with $\beta = 1/2$ will reach
the goal $t$.  However, there is no proper subgraph of $G$ in which 
the agent will reach the goal.
The point is that the agent starts out expecting to follow the
path $s$-$a$-$b$-$t$, but when it gets to node $a$ it finds the
remainder of the path $a$-$b$-$t$
too expensive to justify the reward, and it switches
to $a$-$c$-$t$ for remainder.
With just the path $s$-$a$-$b$-$t$ in isolation, 
the agent would get stuck at $a$; and 
with just $s$-$a$-$c$-$t$, the agent would never start out from $s$.
It is crucial that the agent mistakenly believe the upper path is an
option in order to eventually use the lower path to reach the goal.

It is interesting, of course, to consider real-life analogues of
this phenomenon.  In some settings, the structure in 
Figure \ref{fig:minimal-ex01} could correspond to deceptive practices
on the part of the designer of $G$ --- in other words, inducing the agent
to reach the goal by misleading them at the outset.
But there are other settings in real life where one could
argue that the type of deception represented here is 
more subtle, not any one party's responsibility, and potentially even salutary.
For example,
suppose the graph schematically represents the learning of a skill such as
a musical instrument. There's the initial commitment corresponding to
the edge $(s,a)$, and then the fork at $a$ where one needs to decide
whether to ``get serious about it'' (taking the expensive edge $(a,b)$)
or not (taking the cheaper edge $(a,c)$).
In this case, the agent's trajectory could describe the story of someone
who derived personal value from learning the violin (the lower path) even 
though at the outset they believed incorrectly that they'd be willing to put
the work into becoming a concert violinist (the upper path, with the
edge $(a,b)$ that proved too costly once the agent was standing at $a$).

\xhdr{The Structure of Minimal Motivating Subgraphs}
Given that there is sometimes no single path in $G$ that is motivating,
how rich a subgraph do we necessarily need to motivate the agent?
Let us say that a subgraph $G^*$ of $G$ is a {\em minimal motivating subgraph}
if (i) $G^*$ is motivating, and (ii) no proper subgraph of $G^*$ is motivating.
Thus, for example, in Figure \ref{fig:minimal-ex01},
the graph $G$ is a minimal motivating subgraph of itself; no proper
subgraph of $G$ is motivating.

Concretely, then, we can ask the following question:
what can a minimal motivating subgraph look like?
For example, could it be arbitrarily dense with edges?

In fact, minimal motivating subgraphs necessarily have a sparse
structure, which we now describe in our next theorem.
To set up this result, we need the following definition.
Given a directed acyclic graph $G$ and a path $P$ in $G$,
we say that a path $Q$ in $G$ is a {\em $P$-bypass} if the
first and last nodes of $Q$ lie on $P$, and no other nodes of $Q$ do;
in other words, $P \cap Q$ is equal to the two ends of $Q$.

We now have

\begin{theorem}
If $G^*$ is a minimal motivating subgraph (from start node $s$
to goal node $t$), then it contains an $s$-$t$ path $P^*$ with
the properties that
\begin{itemize}
\item[(i)] Every edge of $G^*$ is either part of $P^*$ or lies on
a $P^*$-bypass in $G^*$; and
\item[(ii)] Every node of $G^*$ has at most one outgoing edge that
does not lie on $P^*$ (in other words, $G^* - E(P^*)$ has maximum
out-degree $1$).
\end{itemize}
\label{thm:minimal}
\end{theorem}

\begin{proof}
If $G^*$ is a motivating subgraph, then the agent in $G^*$ will
reach $t$.  Consider the $s$-$t$ path the agent follows;
we use this as the path $P^*$ in the theorem statement.

We first establish (i).  Consider any edge $(v,w)$ of $G^*$ that
is not on $P^*$.
If there is a path $Q_v$ from $P^*$ to $v$, and a path $Q_w$
from $w$ to $P^*$, then $(v,w)$ must lie on a $P^*$-bypass, by
simply concatenating 
the suffix of $Q_v$ from its last meeting with $P^*$,
then the edge $(v,w)$, and then the prefix of $Q_w$ up to its first
meeting with $P^*$.

So suppose by way of contradiction that
$G^*$ contains an edge $(v,w)$ that does not belong to a $P^*$-bypass.
Then it must be that there is either no path from $P^*$ to $v$,
or no path from $w$ to $P^*$.  In either case, for any node $u$ on $P^*$,
there is no $u$-$t$ path that contains $(v,w)$.  Thus
the shortest-path evaluation of an agent standing at $u$ will not
be affected if $(v,w)$ is deleted.  Since the agent follows $P^*$
from $s$ to $t$, it must be that the agent would also do that
in $G^* - \{(v,w)\}$, and this contradicts the minimality of $G^*$.

Now we consider property (ii).  For any node $w$, we use 
$d^*(w,t)$ to denote the cost of a shortest $w$-$t$ path in $G^*$
(evaluated according to the true costs without present bias).
Also, we fix a topological ordering of $G^*$.
Suppose by way of contradiction that
$(v,w)$ and $(v,w')$ are both edges of $G^*$, where neither $w$ nor
$w'$ belongs to $P^*$.
Let $Q_w^*$ and $Q_{w'}^*$ denote shortest $w$-$t$ and $w'$-$t$ paths
respectively.
Suppose (swapping the names of $w$ and $w'$ if necessary) that
$c(v,w) + d^*(w,t) \leq c(v,w') + d^*(w',t)$.
Moreover, in the event the two sides of this inequality are equal,
we assume the agent uses a consistent tie-breaking rule, so
that if it is ever indifferent between using the edge $(v,w)$ or $(v,w')$,
it chooses $(v,w)$.

We claim that $G^* - \{v,w'\}$ is a motivating subgraph,
which will contradict the minimality of $G^*$.
To see why, consider any node $p \in P^*$ that precedes $v$ in 
the chosen topological ordering and from which the
agent's planned path contains $v$.
We consider two cases: if $p \neq v$, or if $p = v$.
If $p \neq v$, then the planned path cannot contain the edge $(v,w')$, since
$(v,w)$ followed by $Q_w^*$ is no more expensive than $(v,w')$
followed by $Q_{w'}^*$, and if they are equal then we have
established that the agent breaks ties in favor of $(v,w)$ over $(v,w')$.
If $p = v$, then the planned path also cannot contain the edge $(v,w')$,
since the planned path's next edge lies on $P^*$, while $(v,w')$ is not on $P^*$.

Thus there is no node from which the agent's planned path 
contains the edge $(v,w')$.
Finally, we argue the agent will make the same sequence of decisions
in $G^* - \{(v,w')\}$ and $G^*$. 
Indeed, suppose there were a node where the agent made a different decision,
and let $p \in P^*$ be the first such node.
In $G^*$, the agent's decision at $p$ is to follow the edge $(p,q)$ on $P^*$,
as part of a planned path $R_p$.
In $G^* - \{v,w'\}$ the agent won't decide to quit: the path
$R_p$ is still available, since it does not contain $(v,w')$.
And in $G^* - \{v,w'\}$, the agent can't now prefer a path $R_p'$ to $R_p$,
since both $R_p$ and $R_p'$ were available in $G^*$ as well,
and the agent preferred $R_p$.
Thus the agent makes the same sequence of decisions
in $G^* - \{(v,w')\}$ and $G^*$, 
and so $G^* - \{v,w'\}$ is a motivating subgraph.
\end{proof}

\section{Further Directions: Designing for Heterogeneous Agents}
\label{sec:design-heterogeneous}

In the previous section we considered a set of questions that have
a design flavor --- how do we structure a graph to motivate an
agent to reach the goal?
A further general direction along these lines is to consider 
analogous questions for a collection of heterogeneous agents
with different levels of present bias $\beta$.

In particular, suppose we have a population of agents, each with
its own value of $\beta$, and we would like to design a structure
in which as many of them as possible reach the goal.
Or, adding a further objective, we may want many of them to reach the goal
while minimizing the amount of (wasted) work done by agents who 
make partial progress but then fail to reach the goal.

This is a broad question that we pose primarily as a direction 
for further work.  
In particular, it is an interesting open question to explore
the case of motivating subgraphs in the style of the previous
section when there is not just one agent but
a population of agents with heterogeneous 
values of $\beta$.
To illustrate some of the considerations that arise,
we give a tight analysis of a problem with heterogeneous agents
in a model 
that is structurally much simpler than our graph-theoretic formulation.

\xhdr{Partitioning a Task for a Heterogeneous Population}
The design question we consider is to divide up a single task into $k$
pieces so as to make it easier to motivate an agent to perform it.
This is different from the setting of our graph model, in which tasks
were indivisible and hence the splitting of a task wasn't part of
the set of design options.
As we will see, even the problem of dividing a single task already
contains some subtlety, and it lets us think about the relationships
between agents with different levels of present bias.

We model our population of agents as a continuum, 
with $\beta$ uniformly distributed over $[0,1]$.
Thus, rather than talking about the number of agents from a finite set
who complete the task, we will think about the fraction that complete it.

Thus, suppose we have a task of cost $c$ and reward $r$, and it will be
performed by a continuum of agents with $\beta$ uniformly distributed 
over $[0,1]$.  We are allowed to partition the task into $k$ steps,
of costs $c_1, \ldots, c_k$, so that $c = \sum_{i=1}^k c_i$.
Any partition has a {\em completion rate}, defined to be 
the fraction of agents who complete all the steps of the task.
Our goal is to find a partition with maximum completion rate.

We define one additional piece of terminology: 
for any step in the decomposition of the
task, the {\em bottleneck} value of the step
is the minimum $\beta$ for which an agent of present bias $\beta$
will complete that step (when they evaluate
the future steps and the reward as well using their parameter $\beta$).

Let us start by considering some small values of $k$, and then
consider the general case.

\xhdr{${\bf k = 1}$}
When $k = 1$, there is no choice in how to partition, since we
just have a single step of cost $c$.
In this case, the agents that perform the task are those for whom
$c \leq \beta r$, and hence $\beta \geq c/r$.  In other words, 
the bottleneck value of $\beta$ is $c/r$.
We write $\gamma$ for this quantity $c/r$;
in terms of $\gamma$, the completion rate is thus $1 - \gamma$.

\xhdr{${\bf k = 2}$}
When $k = 2$, we must divide the task into two steps so that the
cost of the first step is some $x \in [0,c]$, and the cost of the second
step is $c-x$.  In this case, the bottleneck of the second step is
$(c-x)/r$, and the bottleneck of the first step is the value $\beta$
such that $x + \beta(c - x) = \beta r$, which implies
$\beta = \dfrac{x}{x + r - c}.$
We write $\delta = r - c$, and so we have the bottleneck 
$\beta = \dfrac{x}{x + \delta}.$

Now, the first bottleneck is monotone increasing in $x$, and
the second bottleneck is monotone decreasing in $x$.
We need to choose $x$ to minimize the larger of the two bottlenecks,
which occurs when they are equal.  Hence
\begin{eqnarray*}
\frac{c - x}{r} & = & \frac{x}{x + \delta} 
\end{eqnarray*}
from which we can solve for $x$ to get $x = \sqrt{\delta r} - \delta$.
With this choice of $x$, the bottleneck of the two steps is the same; 
it is 
$$\frac{x}{x + \delta}
  = 1 - \frac{\delta}{x + \delta} 
  = 1 - \frac{\delta}{\sqrt{\delta r}}
  = 1 - \sqrt{\frac{\delta}{r}}
  = 1 - \sqrt{\frac{r - c}{r}}
  = 1 - \sqrt{1 - \gamma},$$
and so the completion rate is $\sqrt{1 - \gamma}$.

\xhdr{Arbitrary values of $k$}
For a fixed value of $r$, let $f_k(c)$ denote the maximum 
completion rate of a task of cost $c$ and reward $r$, when
we can partition it optimally into $k$ steps.
Let $x_k^*(c)$ denote the cost of the first step in 
the optimal solution.

We claim by induction that 
$$f_k(c) = (1 - \gamma)^{1/k}$$ 
and 
$$x_k^*(c) = \delta^{(k - 1)/k} r^{1/k} - \delta,$$
and moreover that in the optimal solution the bottlenecks of all steps 
are the same, and equal to $1 - f_k(c)$.
Note that we have already established the case $k = 2$ above.

For $k > 2$, suppose we give the first step a cost of $x$.
Then the bottleneck for the first step is, as before, $x / (x + \delta)$.
The induction hypothesis tells us that the bottleneck for 
each of the remaining steps is 
$$1 - f_{k-1}(c-x) = 1 - \left(1 - \frac{c - x}{r}\right)^{1/(k-1)}
 = 1 - \left(\frac{x + \delta}{r}\right)^{1/(k-1)}.$$
As in the case $k = 2$, the bottleneck in the first step is
monotone increasing in $x$, and the shared bottleneck in the 
remaining steps is monotone decreasing in $x$.
Thus to minimize the largest bottleneck, we set the first 
bottleneck equal to the shared bottleneck in the remaining steps, obtaining
\begin{eqnarray*}
\frac{x}{x + \delta} & = & 1 - \left(\frac{x + \delta}{r}\right)^{1/(k-1)} 
\end{eqnarray*}
and hence by solving for $x$, we get
$x = \delta^{(k-1)/k} r^{1/k} - \delta$.
With this value of $x$, we get a shared bottleneck of 
$$\frac{x}{x + \delta}
  = 1 - \frac{\delta}{x + \delta} 
  = 1 - \frac{\delta}{\delta^{(k-1)/k} r^{1/k}}
  = 1 - \left(\frac{\delta}{r}\right)^{1/k}
  = 1 - \left(\frac{r - c}{r}\right)^{1/k} 
  = 1 - (1 - \gamma)^{1/k},$$ 
which completes the induction step.

Note also that there is no wasted effort in this solution:
since all the bottlenecks are the same, 
any agent who performs the first task will perform the remaining
$k - 1$ tasks as well.
Thus no agents get partway through the sequence of tasks and then quit.

\section{Conclusion and Open Questions}

We have developed a graph-theoretic model in which an agent
constructs a path from a start node $s$ to a goal node $t$
in an underlying graph $G$ representing a sequence of tasks.
Time-inconsistent agents may plan an $s$-$t$ path that is
different from the one they actually 
follow, and this type of behavior in the model
can reproduce a range of qualitative phenomena including
procrastination, abandonment of long-range tasks, and 
the benefits of a reduced set of options.
Our results provide characterizations for a set of
basic structures in this model, including for graphs achieving the 
highest cost ratios between time-inconsistent agents and
shortest paths, and we have investigated the structure of
minimal graphs on which an agent is motivated to reach the goal node.

There are many specific open questions suggested by this work,
as well as more general directions for future research.
We begin by mentioning three concrete questions.
\begin{enumerate}
\item Our graph-minor characterization in terms of the $k$-fan $\F_k$
is useful for instances with very large cost ratios.  
But this same graph minor structure may be useful for constant cost ratios
as well. In particular, is there a function $f(\cdot)$ such that for
all $k \geq 1$, any instance $G$ in which the undirected skeleton $\sigma(G)$ 
contains no $\F_k$ minor must have cost ratio at most $f(k)$?
\item We have investigated the structure of minimal motivating subgraphs.
But how hard is it computationally to find motivating subgraphs?
Is there a polynomial-time algorithm that takes 
an instance $G$ (including a reward $r$) and determines whether 
$G$ contains a motivating subgraph with reward $r$?
\item The definition of a motivating subgraph is based on a scenario
in which a designer can delete nodes and/or edges from the task graph
to make the goal $t$ easier to reach.  An alternate way to motivate
an agent to reach $t$ is to place intermediate rewards on specific
nodes or edges; the agent will claim each reward if it reaches the
node or edge on which it is placed.  Now the question is to place
rewards on the nodes or edges of an instance $G$ such that the agent
reaches the goal $t$ while claiming as little total reward as possible;
this corresponds to the designer's objective to pay out as little
as possible while still motivating the agent to reach the goal.

There are a couple of points to note in this type of question about
intermediate rewards.  First, one should think about the implication
for creating ``exploitive'' solutions in which the agent is motivated
by intermediate rewards that it will never claim, because these rewards
are on nodes that the agent will never reach.  One could imagine a version
of the problem, for example, in which rewards may only be placed on 
nodes or edges that the agent will actually traverse in the resulting solution.
A second issue is the option of using negative rewards as well as 
positive ones; a negative reward could correspond to a step in which
the agent has to ``pay into the system'' with the intention of 
receiving counterbalancing positive rewards at later points 
in the traversal of $G$.
Note that there is a connection between negative rewards and questions
involving motivating subgraphs:
by placing a sufficiently large negative reward on a particular node
or edge, it is possible to ensure that the agent won't traverse it,
and hence we can use negative intermediate rewards to
implement the deletion of nodes and edges.
\end{enumerate}

Beyond these specific questions, there is also a wide range of broader
issues for further work.
These include finding structural properties beyond our graph-minor
characterization that have a bearing on the cost ratio of a given instance;
obtaining a deeper understanding of the relationship between agents
with different levels of time-inconsistency as measured by different
values of $\beta$;
and developing algorithms for designing graph structures that
motivate effort as efficiently as possible, including for the
case of multiple agents with diverse time-inconsistency properties.


\xhdr{Acknowledgments}
We thank Supreet Kaur, Sendhil Mullainathan, and Ted O'Donoghue
for valuable discussions and suggestions.

\bibliographystyle{plain}
\bibliography{n}

\begin{thebibliography}{10}

\bibitem{akerlof-procrastination}
George~A. Akerlof.
\newblock Procrastination and obedience.
\newblock {\em American Economic Review: Papers and Proceedings}, 81(2):1--19,
  May 1991.

\bibitem{ariely-deadlines02}
Daniel Ariely and Klaus Wertenbroch.
\newblock Procrastination, deadlines, and performance: self-control by
  precommitment.
\newblock {\em Psychological Science}, 13(3):219--224, May 2002.

\bibitem{carstensen-parametric-sp}
P.~Carstensen.
\newblock {\em The complexity of some problems in parametric linear and
  combinatorial programming}.
\newblock PhD thesis, U. Michigan, 1983.

\bibitem{diestel-graph-theory-book}
Reinhard Diestel.
\newblock {\em Graph Theory}.
\newblock Springer, 3 edition, 2005.

\bibitem{edelsbrunner-comp-geom-book}
Herbert Edelsbrunner.
\newblock {\em Algorithms in Combinatorial Geometry}.
\newblock Springer, 1987.

\bibitem{frederick-time-inconsist-surv}
Shane Frederick, George Loewenstein, and Ted O'Donoghue.
\newblock Time discounting and time preference.
\newblock {\em Journal of Economic Literature}, 40(2):351--401, June 2002.

\bibitem{kaur-self-control-aer}
Supreet Kaur, Michael Kremer, and Sendhil Mullainathan.
\newblock Self-control and the development of work arrangements.
\newblock {\em American Economic Review: Papers and Proceedings},
  100(2):624--628, 2010.

\bibitem{laibson-quasi-hyperbolic}
David Laibson.
\newblock Golden eggs and hyperbolic discounting.
\newblock {\em Quarterly Journal of Economics}, 112(2):443--478, 1997.

\bibitem{nikolova-parametric-sp}
Evdokia Nikolova, Jonathan~A. Kelner, Matthew Brand, and Michael Mitzenmacher.
\newblock Stochastic shortest paths via quasi-convex maximization.
\newblock In {\em Proc. 14th European Symposium on Algorithms}, pages 552--563,
  2006.

\bibitem{odonoghue-now-or-later}
Ted O'Donoghue and Matthew Rabin.
\newblock Doing it now or later.
\newblock {\em American Economic Review}, 89(1):103--124, March 1999.

\bibitem{odonoghue-long-term}
Ted O'Donoghue and Matthew Rabin.
\newblock Procrastination on long-term projects.
\newblock {\em Journal of Economic Behavior and Organization}, 66(2):161--175,
  May 2008.

\bibitem{pollak-time-inconsist}
R.~A. Pollak.
\newblock Consistent planning.
\newblock {\em Review of Economic Studies}, 35(2):201--208, April 1968.

\bibitem{russell-norvig-book}
Stuart~L. Russell and Peter Norvig.
\newblock {\em Artificial Intelligence: A Modern Approach}.
\newblock Prentice Hall, 1994.

\bibitem{strotz-time-inconsist}
R.~H. Strotz.
\newblock Myopia and inconsistency in dynamic utility maximization.
\newblock {\em Review of Economic Studies}, 23(3):165--180, 1955.

\end{thebibliography}


\end{document}